\documentclass[11pt,runningheads,a4paper]{llncs}

\usepackage{amssymb}
\usepackage{graphicx}
\usepackage{hyperref}
\hypersetup{
  colorlinks   = true,    
  urlcolor     = blue,    
  linkcolor    = blue,    
  citecolor    = red      
}
\usepackage{url}
\usepackage[hmarginratio=1:1,top=32mm,columnsep=20pt]{geometry} 
\usepackage{multicol} 
\usepackage[hang, small,labelfont=bf,up,textfont=it,up]{caption} 
\usepackage{booktabs} 
\usepackage{float} 
\usepackage{hyperref} 
\usepackage{amssymb,amsmath} 
\usepackage{lettrine} 
\usepackage{paralist} 
\usepackage{verbatim}
\usepackage[ruled,vlined,noline]{algorithm2e}
\usepackage{float}
\usepackage{csquotes}

\DeclareMathOperator*{\argmin}{arg\,min}

\begin{document}
\mainmatter  
\title{Saving Space by Dynamic Algebraization Based on Tree Decomposition: Minimum Dominating Set}
\titlerunning{Saving Space by Dynamic Algebraization: Minimum Dominating Set}
%
%
\author{Mahdi Belbasi, Martin F\"urer%
\thanks{Research supported in part by NSF Grant CCF-1320814.}}
\authorrunning{Mahdi Belbasi, Martin F\"urer}

\institute{Department of Computer Science and Engineering \\
	Pennsylvania State University \\
	University Park, PA 16802,  USA \\
	\{mahdi,furer\}@cse.psu.edu \\
\url{http://www.cse.psu.edu/~furer}}
%
%
\maketitle
\begin{abstract}
An algorithm is presented that solves the Minimum Dominating Set problem exactly using polynomial space based on dynamic programming for a tree decomposition. A direct application of dynamic programming based on a tree decomposition would result in an exponential space algorithm, but we use zeta transforms to obtain a polynomial space algorithm in exchange for a moderate increase of the time. This framework was pioneered by Lokshtanov and Nederlof 2010 and adapted to a dynamic setting by F\"urer and Yu 2017. Our space-efficient algorithm is a parametrized algorithm based on tree-depth and treewidth. The naive algorithm for Minimum Dominating Set runs in $\mathcal{O}^*(2^n)$ time. Most of the previous works have focused on time complexity. But space optimization is a crucial aspect of algorithm design, since in several scenarios space is a more valuable resource than time.\\
Our parametrized algorithm runs in $\mathcal{O}^*(3^{d})$, and its space complexity is $\mathcal{O}(nk)$, where $d$ is the depth and $k$ is the width of the given tree decomposition. We observe that Reed's 1992 algorithm constructing a tree decomposition of a graph uses only polynomial space. So, even if the tree decomposition is not given, we still obtain an efficient polynomial space algorithm. There are some other algorithms which use polynomial space for this problem, but they are not efficient for graphs with small tree depth.
\end{abstract}
\section{Introduction}
Improving the running time of exponential time algorithms for finding exact solutions of NP-complete problems like Minimum Dominating Set (MDS) have been studied for almost half a century. It seems that in the beginning they were emphasizing the running time but gradually some research started working on space optimization as well. In 1977, Tarjan and Trajonowsky \cite{TarjTroj77} improved the running time of Maximum Independent Set from $\mathcal{O}(2^n)$ to $\mathcal{O}(2^{\frac{n}{3}})$ using polynomial space. The Independent Set problem is strongly connected to Dominating Set since an independent set is also a minimal dominating set if and only if it is a maximal independent set. But we should mention that a dominating set is not necessarily an independent set. So we look for the results on the Maximum Independent Set problem as well, since it gives some ideas about the dominating set problem. Then in 1986, Jian improved Tarjan's algorithm and gave a better running time ($\mathcal{O}(2^{0.304n})$) again using polynomial space \cite{Jia86}. There is a trade-off between running time and space complexity. In the mentioned paper Jian states that there is a faster algorithm by Berman ($\mathcal{O}(2^{0.289})$) but this time using exponential space (based on a private discussion between Jian and Berman).\\
Going back to the dominating set problem, in 2005, Grandoni \cite{journals/jda/Grandoni06} gave an algorithm which runs in $\mathcal{O}(1.9053^n)$ using polynomial space. Before this paper, the best algorithm for the dominating set problem was the trivial one which checks all possible subsets. He gives an algorithm with $\mathcal{O}(1.3803^k)$ running time for the minimum set cover problem where $k$ is the dimension of the problem\footnote{In the set cover problem, the dimension is $|S| + |U|$ where $U$ is the universe and $S$ is a collection of subsets of $U$. The goal is to find a minimum-size subset $S' \subseteq S$ that covers all elements in $U$.}. Since the minimum dominating set problem can be converted to a minimum set cover problem by setting $k = 2n$, the running time of this algorithm for MDS is $\mathcal{O}(1.3803^{2n}) = \mathcal{O}(1.9053^n)$. Grandoni \cite{journals/jda/Grandoni06} gives an algorithm with better running time $\mathcal{O}(1.8021^n)$ for MDS. It is also based on converting the dominating set problem to a vertex cover problem but this time using exponential space. For MSC, the later algorithm runs in $\mathcal{O}(1.3424^k)$ and again by setting $k = 2n$ we have the algorithm for MDS with $\mathcal{O}(1.3424^{2n}) = \mathcal{O}(1.8021^n)$ time complexity. Grandoni uses a simple recursive algorithm for MSC and MDS. He uses dynamic programming which saves the answers for the partial problems. We present a framework to convert a dynamic programming algorithm to an algorithm which uses only polynomial space. Our algorithm do not save the results of the smaller problems and it computes them again whenever they are needed.\\
In 2005, Fomin, Grandoni and Kratsch \cite{FomGraKra05a} introduced an approach named \enquote{measure and conquer} based on the \enquote{branch and reduce} approach given by Davis and Putnam \cite{DP60} in 1960.
In the branch and reduce approach, there are some reduction rules  which simplify the problem. Also, there are some branching approaches which are used when we cannot reduce the problem and we have to run it on some other branches. Measure and conquer is mostly used to measure sub-problems more accurately and giving tighter bounds. They chose Grandoni's algorithm mentioned above (with running time $\mathcal{O}^*(2^{0.85n})$) and by choosing a better measuring method for sub-problems, they showed that actually its time complexity is $\mathcal{O}^*(2^{0.598n})$ using exponential space, and also $\mathcal{O}^*(2^{0.610})$ using polynomial space. This shows that in some cases, exponential recursive algorithms are not measured very well and by measuring them accurately we may gain better running time for the same algorithm. They gave a lower bound for the time complexity of this algorithm (not this problem) which is  $\Omega(2^{0.333n})$. \\
Later, in 2008, Van Rooji and Bodlaender \cite{journals/corr/abs-0802-2827} gave two faster algorithms for the dominating set problem using polynomial space. The running time of these algorithms are $\mathcal{O}(1.5134^n)$ and $\mathcal{O}(1.5063^n)$ both using polynomial space. Like previous work by Grandoni \cite{journals/jda/Grandoni06}, they took advantage of the minimum vertex cover problem and formulated MDS as a minimum vertex cover problem. In this paper, they modified the measure and conquer method in a way that it not only a way of analyzes the running time more precisely but also modifies the algorithm by itself. For this purpose, they added a group of rules to the previous rules which tell when the algorithm should be improved. \\
In 2010, Lokshtanov and Nederlof \cite{LokNed2010} gave a framework turning a dynamic programming algorithm for several problems using exponential space to one which uses only polynomial space. They used DFT\footnote{Discrete Fourier transform} and Zeta transform to simplify the complicated convolution operation and convert it to pointwise multiplication which simpler. They gave conditions for turning the dynamic programming algorithms into algorithms with almost the same time complexity using much less space.\\
Then, in 2017 F\"urer and Hu \cite{martin} combined this idea \cite{LokNed2010} with tree decompositions. They took Perfect Matching as a case study. Unlike \cite{LokNed2010}, they handled dynamic underlying set. They introduced the zeta transforms for dynamic underlying sets. We take a similar approach to MDS.
\subsection{Previous works}
We reviewed some of previous works on MDS in the previous section. Here we  gather them in a table and compare their running time and space complexity.
\begin{center}
 \begin{tabular}{||c c c c c||}  
 \hline
 Author(s) & Year & Time Complexity & Space Complexity & Comments\\ 
 [0.5ex] 
 \hline\hline
  & & $\mathcal{O}(2^n)$ &polynomial&\\
  \hline
 Tarjan, and Trajonowsky \cite{TarjTroj77}& 1977 & $\mathcal{O}(2^{\frac{2}{3}})$ & polynomial & MIS \\ 
 \hline
 Grandoni \cite{journals/jda/Grandoni06}& 2005 & $\mathcal{O}(1.9053^n) \simeq \mathcal{O}(2^{0.9301n})$ & polynomial & \\
 \hline
 Grandoni \cite{journals/jda/Grandoni06}& 2005 & $\mathcal{O}(1.8021^n) \simeq \mathcal{O}(2^{0.8497n})$ & exponential & \\
 \hline
 Fomin et. al. \cite{FomGraKra05a}& 2005 & $\mathcal{O}(2^{0.610n}) $ & polynomial & \\
 \hline
 Van Rooji, and Bodlaender \cite{journals/corr/abs-0802-2827}& 2008 & $\mathcal{O}(1.5134^{n}) \simeq \mathcal{O}(2^{0.5978n})$ & polynomial & \\
 \hline
 Van Rooji, and Bodlaender \cite{journals/corr/abs-0802-2827}& 2008 & $\mathcal{O}(1.5063^{n}) \simeq \mathcal{O}(2^{0.5911n})$ & polynomial & \\
 \hline
\end{tabular}
\end{center}

\section{Notations}
In this section we  recall the problem and mention notations we use later. Some of the definitions in this section are based on the notation used in the \enquote{Parametrized Algorithms} \cite{cygan2015parameterized} book. 
\begin{definition}{\textbf{Closed Neighbourhood of a Subset of Vertices}}
    Let $X$ be a subset of vertices of a given graph $G = (V, E)$. Then the closed neighbourhood of $X$ is defined as below:
    \[
        N[X] = \cup_{v\in X}N[v],
    \]
    where $N[v]$ is the closed neighbourhood of $v$.
\end{definition}
\begin{definition}{\textbf{Dominating Set:}} 
A subset of vertices $D$ is a dominating set of a given graph $G = (V, E)$ if $N[D] = V$.
\end{definition}

\subsection{Tree Decomposition}
For a given graph $G = (V, E)$, a \emph{tree decomposition} of $G$ is a tree $\mathcal{T} = (V_{\mathcal{T}}, E_{\mathcal{T}})$ such that each node $x$ in $V_{\mathcal{T}}$ is associated with a set $B_x$ (called the bag of $x$) of vertices in $G$, and $\mathcal{T}$ has the following properties:
\begin{itemize}
    \item The union of all bags is equal to $V$. In other words, for each $v\in V$, there exists at least one $x\in V_{\mathcal{T}}$ containing $v$.    
    \item For every edge $\{u, v\} \in E$, there exists a node $x$ such that $u,v\in B_x.$     
    \item For any nodes $x, y \in V_{\mathcal{T}}$, and any node $z\in V_{\mathcal{T}}$ belonging to the path connecting $x$ and $y$ in $\mathcal{T}$, $B_x \cap B_y \subseteq B_z.$
\end{itemize}
The \emph{width of a tree decomposition} is the size of its largest bag minus one. The \textit{treewidth} of a graph $G$ is the minimum width over all tree decompositions of $G$ called $tw(G)$. We use the letter $k$ for the treewidth of the graph we are working on. In 1987, Arnborg et al.\ showed that constructing a tree decomposition with the minimum treewidth is an NP-hard problem \cite{ArnCorPro-SJADM-87}. In our case we don't need the optimal tree decomposition, as a near optimal tree decomposition also works. A linear time algorithm for finding the minimum treewidth of a given graph with a bounded (by a constant) treewidth is given in \cite{bodlaender1996linear}. 
Bodlaender et al.\ \cite{bodlaender1995approximating} gave an 
$\mathcal{O}(\log n)-$approximation algorithm. The improved $\mathcal{O}(\log k)-$approximation algorithm can be found in \cite{amir2001efficient,amir2010approximation,bouchitte2004treewidth}. The result has been further improved to $\mathcal{O}(\sqrt{\log k})$ in \cite{feige2008improved}.
It is well known that every tree decomposition can be modified in polynomial time into a nice tree decomposition with the same width and size $\mathcal{O}(n)$.\\
We use the notion of a {\em modified nice tree decomposition} which is defined in \cite{martin} with a small modification. Instead of introducing all edges in one node, we introduce them one by one. This helps us understand the procedure better. We also require leaf nodes to have empty bags.\\
We use $T(G)$ or simply $T$ to denote a nice tree decomposition of a graph $G$. \\
For a node$x$ in $T$, let $T_x$ be the subtree rooted at $x$. \\
We define $\tau_x = (V_x, E_x)$ to be the graph with $V_x$ being the union of all bags in $T_x$ and $E_x$ being the set of all edges introduced in $T_x$. With the traditional method of using the subgraphs induced by the already introduced vertices, i.e., by automatically including all edges, we would obtain an $O^*(k^d n)$ running time which is not desired.
\\
We take advantage of solving partial problems with different sizes and then solve the main problem according to partial ones. Inspite of previous works like \cite{martin}, here for solving partial problems, we don't have only two sets. We maintain a partition of $B_x$ into three sets:
\begin{itemize}
    \item \textbf{\enquote{Dominating (Completed) set}:} In each partial problem, this set is the set of vertices which are in the dominating set. We denote this set by $C$ (\textbf{C}ompleted).
    \item \textbf{\enquote{Dominated set}:} In each partial problem, this set is the set of vertices which are not contained in the dominating set $C$, and are dominated by vertices in $C$. We denote this set by $D$ (\textbf{D}ominated).
    \item \textbf{\enquote{Waiting set}:} In each partial problem, this set is the set of vertices which are not contained in the dominating set $C$, and are not dominated by it (neither dominating nor dominated ``yet''). We denote this set by $W$ (\textbf{W}aiting to be dominated).
\end{itemize}

For each $x \in T$ (with the bag $B_x$), we define a function $f: B_x \rightarrow \{1, -1, 0\}$. This functions tells for each vertex in $B_x$, in which of the sets $C, D,$ or $W$ it is. If $v\in B_x$, then:
\begin{itemize}
    \item If $f(v) = 1$, it means that $v\in C$.
    \item If $f(v) = -1$, it means that $v\in D$.
    \item If $f(v) = 0$, it means that $v\in W$.
\end{itemize}
We follow a down-to-top approach. It means we start from the leaves and in each step, first we solve the problem for the child(ren) and then we use this computation for the parent node.\\
The underlying reason why we have the third option (set $W$) is that maybe some of the vertices of a bag which are in $W$ would be dominated in upper nodes (in down-to-top approach). For a function $f$ over $B_x$, let's define $c[x, f]$ to be minimum size of set $C \subseteq V_x$ such that:
\begin{itemize}
    \item $D \cap B_x = f^{-1}(1)$.
    \item All vertices of $B_x\setminus f^{-1}(0)$, should be either in $C$ or adjacent to one in $C$. It means, all vertices which don't have $f(v) = 0$, are dominated by some vertex or dominating some vertex (or vertices).
\end{itemize}

If there is no minimum set $C$ associated with $x$, then we  assign $c[x,f] = + \infty$. Our goal is to compute $c[r, \emptyset]$, where $r$ is the root of our tree decomposition $T$. We start from the leaves and handle each node by handle its child (or children) first. Before introducing recursive formulas for each kind of nodes in $T$, let's define a specific version of $f$: Suppose $X \subseteq V(G)$, and $f: X \rightarrow \{1,-1,0\}$, we define $f_{v\rightarrow p}X\cup \{v\} \rightarrow \{1,-1,0\}$ in this way:

\begin{equation}
    f_{v\rightarrow p}(x)= 
\begin{cases}
    f(x),& \text{if } x \neq v\\
    p    & \text{if } x   =  v
\end{cases}
\end{equation}

Also, we use the restriction of $f$ to some subset like $Y \subseteq X$ as follows: For a function $f$ over $X$, we denote the restriction of $f$ to $Y$ by $f|_Y$.\\

\begin{definition}
    For a universe $\mathcal{V}$ and a ring $\mathcal{R}$, the set $\mathcal{R}[2^{\mathcal{V}}]$ is the set of all function $f: 2^{\mathcal{V}} \rightarrow \mathcal{R}$.
\end{definition}

Now we review the definition of the Zeta and M\"obius Transforms.\\
\textbf{Zeta Transform}: The \emph{zeta transform} of a function $f \in \mathcal{R}[2^{\mathcal{V}}]$ is defined by:
\begin{equation}
    \zeta f[Y] = \sum_{X \subseteq Y} f[X].
\end{equation}
\\
\textbf{M\"obius Transform/Inversion }: The \emph{M\"obius transform} of a function $f \in \mathcal{R}[2^{\mathcal{V}}]$ is defined to be:
\begin{equation}
    \mu f[Y] = \sum_{X \subseteq Y} (-1)^{Y\setminus X} f[X].
\end{equation}

We will see later that in the recursion formula for join nodes we need the union product \cite{DBLP:journals/corr/abs-cs-0611101}.
\begin{definition}
    Given $f, g \in \mathcal{R}[2^{\mathcal{V}}]$ and $X \in 2^{\mathcal{V}}$, the \emph{Union Product} of $f$ and $g$ which is shown by $(f*_{u}g)$ is defined as:
    \begin{equation}
        (f*_{u}g)[X] = \sum_{X_1\cup X_2 = X}f(X_1)g(X_2).
    \end{equation}
\end{definition}
\section{Saving Space Using Zeta and M\"obius Transforms}
In 2010, Lokshtanov and Nederlof \cite{LokNed2010} used the zeta and M\"obius transforms to save space on some problems like the unweighted Steiner Tree problem. Then in 2017, F\"urer and Yu \cite{martin} also used this approach in a dynamic setting, based on tree decompositions for the Perfect Matching problem.\\
Before going into the depth, let's review a few theorems.
\begin{theorem}
\label{thm:mobzet}(\cite{rota1964foundations,stanley1997enumerative})
   The M\"obius transform is the inverse transform of the zeta transform, i.e. given $f\in \mathcal{R}[2^{\mathcal{V}}]$ and $X \in 2^{\mathcal{V}}$, $\mu(\zeta f)[X] = \zeta(\mu f)[X] = f[X].$
\end{theorem}
\begin{theorem}
\label{thm:pointwise-multiplication}(\cite{DBLP:journals/corr/abs-cs-0611101})
    Applying the zeta transform on a union product operation, will result in a simple point-wise multiplication of zeta transform outputs, i.e. given $f, g \in \mathcal{R}[2^{\mathcal{V}}]$ and $X \in 2^{\mathcal{V}}$, 
    \begin{equation}
        \zeta(f*_{u}g)[X] = (\zeta f) \odot (\zeta g) [X],
    \end{equation}
    where the operation $\odot$ is the point-wise multiplication.
\end{theorem}

All of the previous works which used either DFT or zeta transform (such as \cite{martin}, and \cite{LokNed2010}), have one common central property that the recursive formula in join nodes, can be expressed as a formula using a union product operation. The union product and the subset convolution are complicated operations in comparison to point-wise multiplication or point-wise addition. That is why taking the zeta transform of a formula having union products makes computations easier. As noted earlier in Theorem~\ref{thm:pointwise-multiplication}, taking the zeta transform of such a term, will result in a term having point-wise multiplication which are much easier to handle than the union product. After doing a computation over the zeta transformed values, we can apply the M\"obius transform on the outcome to get the main result (based on Theorem~\ref{thm:mobzet}). In other words, the direct computation has a mirror image using the zeta transforms of the intermediate values instead of the original ones. While the direct computation keeps track of exponentially many intermediate values, the computation over the zeta transformed values partitions into exponentially many branches and they can be executed one after another. We show later that this approach only moderately affects the exponential factor in the time complexity, while improving the space complexity using only polynomial space instead of exponential space.
\subsection{Counting the dominating sets of all sizes}
For each node (of any kind) we introduce a polynomial to compute the number of the dominating sets of all sizes as follows:
\begin{equation}
   P_x^C [D] = \sum_j a_{x,j}^C[D] \, y^j,
\end{equation}
where $a_{x,j}^C[D]$ is the number of the dominating sets $C^*$ of size $j$ in $\tau_x$ with 
$C^* \cap B_x = C$, and where $D$ is the set of all 
dominated vertices (excluding vertices in $C$) in $B_x$. We apply the zeta transform on the coefficients of this polynomial.\\
Now, we show how to compute $a_{x,j}^C[D]$.\\
\begin{itemize}
    \item \textbf{Leaf node.} Assume $x$ is a leaf node, then $B_x = \emptyset$. So,
    $a_{x,j}^{\emptyset}[\emptyset] = 0$ 
    for all $j \neq 0$ since we do not have $j$ vertices in the bag and $a_{x,0}^{\emptyset}[\emptyset] = 1$ since there is an empty dominating set of size  $0$ . This implies that:
    \begin{equation}
        (\zeta a_{x,j}^{\emptyset})[\emptyset] = 
        \begin{cases}
                0  & \text{if } j \neq 0,\\
                1    & \text{if } j = 0.
            \end{cases}
    \end{equation}        

    \item \textbf{Introduce vertex node.} Assume $x$ is an introduce vertex node and let $x'$ be the child of $x$ such that $\exists v \in B_{x}, v\notin B_{x'}: B_x = B_{x'} \cup \{v\}$. Then, if $v \notin D$, dominating sets of any size would be the same for both subtrees ($\tau_x$ and $\tau_{x'}$). Otherwise ($v\in D$), there would be no dominating set through $\tau_x$ to dominate $D$ since $v$ is isolated through $\tau_x$ which means $a_{x,j}^{C}[D] = 0$.
    \begin{equation}
       a_{x,j}^{C}[D] =  
        \begin{cases}
            a_{x',j}^{C\setminus \{v\}}[D]  & \text{if } v \notin D,\\
            0    & \text{if } v   \in  D.
        \end{cases}
    \end{equation}
    Now by applying the zeta transform on the equation above, if $v \in D$, then $(\zeta a_{x,j}^{C})[D] = 0$ and otherwise:
    \begin{equation}
       (\zeta a_{x,j}^{C})[D] = \sum_{Y \subseteq D} a_{x,j}^{C}[Y] = \sum_{Y \subseteq D} a_{x',j}^{C\setminus \{v\}}[Y] = (\zeta a_{x',j}^{C\setminus \{v\}})[D].
    \end{equation}    
    
    \item \textbf{Forget node.} Assume $x$ is a forget node and let $x'$ be the child of $x$ such that $\exists v \in B_{x'}: B_x = B_{x'} \setminus \{v\}$. Then, $v$ should be either in $D$ or $C$, so:
    \begin{equation}
        a_{x, j}^{C}[D] = a_{x', j}^{C}[D \cup \{v\}] + a_{x', j}^{C\cup \{v\}}[D].
    \end{equation}
    Now, by applying zeta transform we have:
    \begin{equation}
    \begin{split}
       (\zeta a_{x, j}^{C})[D] & = \sum_{Y\subseteq D}a_{x, j}^{C}[Y] = \sum_{Y\subseteq D}(a_{x', j}^{C}[Y\cup \{v\}] + a_{x', j}^{C\cup \{v\}}[Y]) = \\& 
       ((\zeta a_{x', j}^{C})[D\cup \{v\}] - (\zeta a_{x', j}^{C})[D]) + ((\zeta a_{x', j}^{C\cup \{v\}})[D]).
    \end{split}
    \end{equation}    
    \item \textbf{Join node.} Assume that $x$ is a join node and let $x'$ and $x''$ be the children of $x$ and also we know that $B_{x} = B_{x'} = B_{x''}$.\\
    Then, if a vertex $v \in B_x$ is in dominating set ($C$), then it should be in dominating set in $B_{x'}$ and $B_{x''}$ as well and vice versa ($C = C' = C''$, where $C'$ and $C''$ are the dominating sets through $\tau_{x'}$ and $\tau_{x'}$ respectively intersecting $B_{x'}$ and $B_{x''}$). By the way, if a vertex is dominated through $\tau_{x'}$ or $\tau_{x''}$ (or both), then it is dominated through $\tau_{x}$. So, $D = D' \cup D''$, where $D'$ (and $D''$) is the dominated set through $\tau_{x'}$ (and $\tau_{x'}$) intersecting $B_{x'}$ (and $B_{x''}$).\\
    In order to compute $a_{x,j}^{C}$ based on pre-computed $a_{x', k}^{C'}$ and $a_{x'', l}^{C''}$ values, we have the recursion below:
    \begin{equation}
       a_{x,j}^{C}[D] = \sum_{j'>1,\, j''>1 \atop j' + j'' = j + |C|}\sum_{D' \cup D'' = D} a_{x',j'}^{C}[D']\cdot a_{x'',j''}^{C}[D''] =  \sum_{j'>1,\, j''>1 \atop j' + j'' = j + |C|}(a_{x',j'}^{C} *_{_u} a_{x'',j''}^{C})[D].
    \end{equation}
    By applying zeta transform, we will have the equation below:
    \begin{equation}\label{eq:join}
       \zeta (a_{x,j}^{C})[D] = \sum_{j'>1,\, j''>1 \atop j' + j'' = j + |C|}(\zeta a_{x',j'}^{C})[D] \cdot (\zeta a_{x'',j''}^{C})[D]
    \end{equation}    
    \item \textbf{Introduce edge node.} Assume that $x$ is an introduce edge node introducing edge $\{u,v\}$. 
    \begin{definition} An \emph{auxiliary leaf node} $x$, is a leaf node which has some vertices in its bag(at least $2$) and is used to introduce an edge (regularly leaf nodes have empty bag). It has only one introduced edge through $B_x$. We use auxiliary leaf nodes to introduce an edge. This helps us to handle introduce edge nodes.
    \end{definition}
    We convert the introduce edge node to a join node such that $x$ has two children $l$ and $x'$ where $B_x = B_l = B_{x'}$ but edges are different. In auxiliary leaf node $l$, we will have only $\{u, v\}$ edge introduced and other vertices are isolated through $\tau_l$ which has only one node ($l$ itself). On the other hand, all other edges are present in $x'$ except $\{u, v\}$. We saw how to handle join nodes, so we can handle introduce edge nodes as well. Remember, this branching does not affect our running time since branching is not balanced and in such a branch we have only one auxiliary leaf node which can be done very easily and has no children. Now, let's see how to handle auxiliary leaf node.
    \item \textbf{Auxiliary leaf node.} Assume $x$ is an auxiliary leaf node with only $\{u,v\}$ edge, then we have four cases. In all cases, $a_{x,j}^C[D] = 0$ if $j \neq |C|$, so we do not consider this condition and assume that $j=|C|$ in all cases:
    \begin{itemize}
        \item If $D = \emptyset$, then,
        \begin{equation}
            a_{x, j}^{C}[\emptyset] = 1
        \end{equation}
        Any given set $C$, dominates an empty set $D$. As a reminder, if $j \neq |C|$, $a_{x,j}^C[D]$ is zero and we do not repeat this case from so on and only consider the cases where $j = |C|$.
        
        \item If $D = \{u, v\}$, then $a_{x,j}^{C}[D] = 0$ since dominating $u$ forces $v$ to be in the dominating set and on the other hand, dominating $v$ forces $u$ to be in the dominating set. So it means that $D = \{u, v\} \subseteq C$ but we know that $D$ and $C$ are disjoint sets so it is impossible. Therefore, there is no dominating set $C$ to dominate $D$.          
        \item If $D$ is either $\{u\}$ or $\{v\}$, then,
        \begin{equation}
            a_{x,j}^{C}[D] =
            \begin{cases}
                1 & \text{if } \{u, v\}\setminus D \subseteq C,\\
                0 & \text{otherwise.}
            \end{cases}
        \end{equation}
        Without loss of generality, assume that $D = \{u\}$, then $C$ can be any subset of $B_x$ including $v$ (the only way to dominate $u$) and excluding $u$ ($D$ and $C$ are disjoint sets).   
        \item And finally if $D\setminus \{u, v\} \neq \emptyset$, then we cannot dominate elements of $D\setminus\{u, v\}$ because they are isolated and they themselves cannot be in $C$ (remember we assumed $C$ and $D$ are distinctand if they are not distinct, then there no let $a_{x,j}^C[D]$ to be zero). So, $a_{x,j}^{C}[D] = 0$. 
    \end{itemize}   
    Now, let us compute the zeta transform of the above cases:
    \begin{itemize}
        \item In case $1$:
        \begin{equation}
            (\zeta a_{x, j}^{C})[D] = (\zeta a_{x, j}^{C})[\emptyset] = a_{x, j}^{C}[\emptyset] = 1                
        \end{equation}
        
        \item In case $2$: 
        \begin{equation}
        \begin{split}
            (\zeta a_{x, j}^{C})[D] &= a_{x, j}^{C}[D] +  a_{x, j}^{C}[\{u\}] + a_{x, j}^{C}[\{v\}] + a_{x, j}^{C}[\emptyset] \\= & 
            \begin{cases}
                0 + 1 + 0 + 1 = 2 & \text{if }   C = \{v\},\\
                0 + 0 + 1 + 1 = 2& \text{if }   C = \{u\},\\
                0 + 0 + 0 + 1 = 1& \text{otherwise}. 
            \end{cases}
        \end{split}
        \end{equation}
        
        \item In case $3$:
        \begin{equation}
        \begin{split}
            (\zeta a_{x, j}^{C})[D] &= a_{x, j}^{C}[D] +  a_{x, j}^{C}[\emptyset] \\=& 
            \begin{cases}
                1 + 1 = 2  & \text{if } \{u, v\}\setminus D \subseteq C,\\
                0 + 1 = 1 & \text{otherwise}.
            \end{cases}
        \end{split}
        \end{equation}
        \item In case $4$: 
        \begin{equation}
            (\zeta a_{x, j}^{C})[D] = (\sum_{Y\subseteq D: Y\setminus \{u, v\} \neq \emptyset} a_{x,j}^{C}[Y]) + (\zeta a_{x,j}^{C})[D\cap \{u, v\}] = (\zeta a_{x,j}^{C})[D\cap \{u, v\}].
        \end{equation}
        which is one of the above cases.
    \end{itemize}
    \end{itemize}
By looking at Eq.~\ref{eq:join}, we see that we have $n-$fold branching in join nodes which is not desired. When computing an $a_{x,j}^C[D]$ for a fixed $C$ and $D$, we always want to compute the whole vector for all $j$ and store it. So for all fixed $C$ and $D$, we store an array $Res[1...n]$ of size $n$, where $Res[j] = a_{x,j}^C[D]$. Now, even though we have $n$ branches but we do not need to recompute them each time. Instead, we can compute them once and store them in array of size $n$ which is linear. We have only one copy of this array because we handle any fixed $C$ and $D$ one after another.
\subsection{Finding The Minimum Dominating Set}
In this section we are going to find the minimum dominating set for a given graph $G= (V, E)$.\\
As mentioned before, assume we have a modified nice tree decomposition $T$ with the root $r$ where $B_r = \emptyset$. Our goal is to compute  \[\argmin_{j = 0}^{n}\{a_{r, j}^{\emptyset}[\emptyset] > 0\},\]where $n = |V(G)|$.\\
We start from $j = 0$ and increment $j$ while $a_{r, j}^{\emptyset}[\emptyset] = 0$.
\begin{algorithm}
 \SetKwInOut{Input}{Input}
 \SetKwInOut{Output}{Output}
 \Input{a modified nice tree decomposition $\tau$ with root $r$ of a graph $G = (V, E)$.}
 \Output{Size of a minimum dominating set of $G$.//$(\zeta a)(x, D, C, j)$ represents $(\zeta a_{x, j}^{C})[D]$.}
 //$D$ is a subset of $B_x\setminus C$ and $0\leq j \leq n$, where $n = |V(G)|$\;
 $j \leftarrow 0$\;
 initialize $Res[1...n]$ to be a zero vector.\\
 //$f(x,D,C, Res, j)$ \text{is computed by Algorithm~\ref{alg:2}.}\\
 \While{$f(r, \emptyset, \emptyset, Res, j) = 0$}{
  $j \leftarrow j+1$\;
 }
 return $j$\; 
 \caption{Minimum dominating sets of all sizes on a modified nice tree decomposition}
 \label{alg:1}
\end{algorithm}

\begin{theorem}
    The Algorithm~\ref{alg:1}, outputs the size of the minimum dominating set correctly.
\end{theorem}
\begin{proof}
    First of all, Algorithm~\ref{alg:1} starts from $j = 0$ and checks if there a minimum dominating set to dominate everything in the tree by increasing $j$ by $1$ step by step. In the while loop, it calls Algorithm~\ref{alg:2}, to compute $f(x,D,C,j)$, where it is an alias $(\zeta_{x,j}^C)[D]$.\\
    The while loops will end as soon as the Algorithm~\ref{alg:1} finds the first $j$ to make $f(x,D,C,j)$ non-zero. So as we explained before, $f(x, D, C, j)$ (or  $(\zeta_{x,j}^C)[D]$), it is a function to compute the number of possible dominating sets of size $j$ in the sub-tree rooted at $x$. Here $x=r$ and the sub-tree is our original tree. So if $j>0$, then it means we have a dominating set for the whole tree of size $j$ and the Algorithm~\ref{alg:1} will return it.\\
    Now, let's see how Algorithm~\ref{alg:2} computes $f(x, D, C, j)$.\\
    This algorithm gets a subtree rooted at $x$, a set  $D\subseteq B_x$ that we want to dominate, a subset $C$ to be the intersection of the whole dominating set in $\tau_x$ (which dominated $D$) named $C^*$ with the bag of $x$.

        \begin{algorithm}[H]
 \SetKwInOut{Input}{Input}
 \SetKwInOut{Output}{Output}
 \Input{a sub-tree of a modified nice tree decomposition $\tau_x$ rooted at $x$, a set $C \subseteq B_x$ to be the intersection of dominating set of $D$ through $\tau_x$ with $B_x$, the vector of all $a_{x,j}^C[D]$ for any fixed $C$ and $D$, and size of the dominating set $C$ to dominate $D$ through $\tau_x$.}
 \Output{$f(x, D, C, Res, j)$. //$f(x, D, C, Res, j)$ represents $(\zeta a_{x, j}^{C})[D]$.}
  \textbf{If }$j \neq |C|$\\
  \hspace{5mm}$Res[j] \leftarrow 0$\;
  \hspace{5mm} return $0$\;
  \textbf{If }$x$ is a leaf node\\
  \hspace{5mm}return $1$\;
  \textbf{If }$x$ is an introduce vertex node //introducing vertex $v$\\
    \hspace{5mm}\textbf{If} $v \in D$:
    return $0$\;
    \hspace{5mm}set $Res[1...n] = [0...0]$\;
    \hspace{5mm}return $f(x', D, C\setminus\{v\}, Res, j)$ //$x'$ is the child of $x$\;
  \textbf{If }$x$ is a forget node //forgetting vertex $v$\\
  \hspace{5mm}return $f(x', D\cup\{v\}, C, j)-f(x', D, C, Res, j)+f(x', D, C\cup\{v\}, Res, j)$ //$x'$ is the child of $x$.\;
  \textbf{If }$x$ is a join node\\
  \hspace{5mm}return $\sum_{j',j''>1 \atop j'+j'' = j-|C|}f(x', D, C, j')f(x'', D, C, j'')$ //$x'$ and $x''$ are the children of $x$.\;
  \textbf{If }$x$ is an auxiliary leaf node //having $\{u,v\}$ as its only edge\\
  \hspace{5mm}\textbf{If} $D = \emptyset$:\\
  \hspace{5mm}\hspace{5mm}\textbf{If } $j  = |C|$:\\
  \hspace{5mm}\hspace{5mm}\hspace{5mm}return $1$\;
  \hspace{5mm}\hspace{5mm}return $0$\;
  \hspace{5mm}\textbf{If} $D = \{u, v\}$:\\
  \hspace{5mm}\hspace{5mm} return $0$\;
  \hspace{5mm}\textbf{If} $D = \{u\}$ or $D = \{v\}$:\\
  \hspace{5mm}\hspace{5mm}\textbf{If } $j = |C|$ and $\{u, v\}\setminus D = C$:\\
  \hspace{5mm}\hspace{5mm}\hspace{5mm} return $2$\;
  \hspace{5mm}\hspace{5mm}\textbf{If } $j = |C|$ and $\{u, v\}\setminus D \neq C$:\\
  \hspace{5mm}\hspace{5mm}\hspace{5mm} return $1$\;
  \hspace{5mm}\hspace{5mm} return $0$\;
  \hspace{5mm} return $f(x, D\cap \{u, v\}, C, j)$\;
 \caption{Computing $f(x, D, C, j)$ (which actually is $(\zeta a_{x,j}^C)[D]$).}
 \label{alg:2}
\end{algorithm}

    As we explained before, keeping track of the number of such $C^*$s of size $j$ will lead us a complicated convolution operation in join nodes. Therefor, we will keep track of the zeta transform of all computation instead on the original ones. This will save space for us. Assume the tree decomposition is given, we will make a copy of it to  save $(\zeta a_{x,j}^C)[D]$ instead of $a_{x,j}^C[D]$. The procedure of computing zeta transforms $f(x, D, C, j)$, is excatly as shown in previous section.
  
    In this step, the algorithm checks the type of $x$. Here we explain the cases:\\
    \begin{itemize}
        \item $x$ is a \textbf{leaf} node: Then the only possible $D$ is an empty set and $C^* = \emptyset$ will dominate it. So there is one way to dominate it and $f(x, D, C, j) = 1$.
        \item $x$ is an \textbf{introduce vertex} node: Suppose $x'$ is the only child of $x$ and $B_x = B_{x'}\cup \{v\}$. As we showed previously $(\zeta a_{x,j}^C)[D] = (\zeta a_{x',j}^{C\setminus\{v\}})[D]$ if $v \notin D$. Otherwise it would be zero.
        \item $x$ is a \textbf{forget} node: Suppose $x'$ is the only child of $x$ and $B_x \cup \{v\} = B_{x'}.$ As shown before,  $((\zeta a_{x', j}^{C})[D\cup \{v\}] - (\zeta a_{x', j}^{C})[D]) + ((\zeta a_{x', j}^{C\cup \{v\}})[D])$.
        \item $x$ is a \textbf{join} node: Suppose $x'$ and $x''$ are the children of $x$ and $B_x  = B_{x'} = B_{x''}.$ As shown before,  $(\zeta a_{x, j}^C)[D] = \sum_{p=0}^{j}(\zeta a_{x',p}^{C})(\zeta a_{x'',j-p-|C|}^{C}) [D]$.
        \item $x$ is a \textbf{auxiliary leaf} node: There are four cases (In all cases if $j\neq |C|$, then $a_{x,j}^C[D] = 0$):
        \begin{itemize}
            \item Case 1: $D = \emptyset:$\\
            \begin{equation}
                (\zeta a_{x,j}^C)[D] = 1.
            \end{equation}
            \item Case 2: $D = \{u, v\}:$\\
            \begin{equation}
                (\zeta a_{x,j}^C)[D] = \begin{cases}
                    2 & \text{if } C = \{v\} \text{ or } \{u\},\\
                    1 & \text{Otherwise.}
                \end{cases}
            \end{equation}
            \item Case 3: D = \{u\} or \{v\}:\\
            \begin{equation}
                (\zeta a_{x,j}^C)[D] = \begin{cases}
                    2 & \text{if } \{u, v\}\setminus D \subseteq C,\\
                    1 & \text{Otherwise.}
                \end{cases}
            \end{equation}
            \item Case 4: $D \setminus \{u, v\} \neq \emptyset:$\\
            \begin{equation}
                (\zeta a_{x,j}^C)[D] = (\zeta a_{x,j}^C)[D\cap \{u, v\}].
            \end{equation}
        \end{itemize}
    \end{itemize}
    Starting from leaf nodes, we compute only one strand at a time and continue the path to the root and we showed how parents can be computed based on their child(ren).\\
\end{proof}
\section{Analyzing Time and Space Complexity of the Algorithm}

In the previous sections, we talked about our algorithm. Now, it is the time to explain the time and the space complexity of the algorithm given above. As we explained beforehand, we have an exact algorithm using only polynomial space.

\begin{theorem}
    Given a tree decomposition $\tau$ of a graph $G$, the Algorithm~\ref{alg:1}, outputs the size of minimum dominating set in  $\mathcal{O}(n^23^{d})$  time using $\mathcal{O}(nk)$ space, where $k$ is the treewidth and $d$ is the tree-depth.
\end{theorem}
\begin{proof}
We know that the size of any dominating set of a Graph $G=(V,E)$ is at most $|V| = n$. So, the while loop in Algorithm~\ref{alg:1} runs at most $n$ times ($j$ is the size of minimum dominating set). In each iteration, we call function $f$ recursively. This function is computed by Algorithm~\ref{alg:2}. So, we need to look at the running time of the Algorithm~\ref{alg:2}:\\
In each node of the tree decomposition, we have at most three branches and it happens at forget nodes. On the other hand, we have $\mathcal{O}(n)$ forget nodes in the tree. Thus, the total running time of Algorithm~\ref{alg:2} is $\mathcal{O}(n3^{d})$. So finally, the total running time of Algorithm~\ref{alg:1} is $\mathcal{O} (n^23^d) = \mathcal{O}^{*}(3^d)$.\\
Space complexity: We do not store the intermediate values and only compute one strand at a time, thus the space complexity is $\mathcal{O}(nk)$.
\end{proof}
\section{Extension}
So far, we saw how to find a minimum dominating set when the tree decomposition is given. Now, we describe how to find the size of a minimum dominating set using polynomial space when the graph itself is given.

\begin{remark}\label{rem:1}(\cite{reed1992finding})
In 1992, B. Reed  gave an approximated vertex separator algorithm that gives another algorithm determining whether a given graph $G$ has a tree decomposition of width at most $k$. Then it finds such a tree decomposition if it exists. The running time of this algorithm is $\mathcal{O}(n\log n)$ for any fixed $k$ using polynomial space.
\end{remark}
 It is not mentioned that the space complexity of the algorithm is polynomial but this algorithm saves only the results for each strand and has a polynomial space complexity.\\
 So, even if the tree decomposition is not given, still we can use Reed's algorithm to find a tree decomposition with width of $k$. We can start from $k=n$ and set $k$ to $\frac{k}{2}$ (binary search) if this algorithm says there is a tree decomposition of width $k$. At some point, the answer would be no and then the algorithm will output the tree decomposition of size $k$. Now, we can use the output to compute the size of minimum dominating 
set of the original graph by Algorithm~\ref{alg:1}.\\
We also should mention that, the same framework can be applied on the Maximum Independent Set (MIS) problem to covert a dynamic programming algorithm using exponential time to a parametrized algorithm taking tree depth as its parameter using only polynomial space. On the other hand, if we look closer to the Maximum Independent Set problem and write down the recursions, we see that also by applying zeta transform we save space, but we sacrifice huge amount of time. But the good point about MIS is that, if we solve it using recursion on tree decomposition, it has already polynomial space complexity and we don't need to apply zeta transform.\\
One good problem would be to see if this framework works on the Hamiltonian Cycle problem. This one is much complicated because the nature of the problem is different. In previous problems, only pairs of vertices and edges (individually) were important but here one should look for all of possible disjoint paths. \\
In general, any problem of graphs (even if the problem is not on graph, (maybe) we can convert it) which has convolution in the join nodes, can be a candidate for our framework.

\section{Conclusion}
We applied the dynamic algebraization approach (\cite{martin}) to the Minimum Dominating Set problem to give a space-efficient dynamic programming algorithm using polynomial space. Our algorithm runs in time $\mathcal{O}(n^23^d)$ using $\mathcal{O}(nk)$ space where $d$ and $k$ are depth and width of the tree decomposition respectively. Even if the tree decomposition is not given, as mentioned in Remark~\ref{rem:1}, we can construct a tree decomposition with sufficiently small treewidth by using polynomial space and then apply our algorithm. Again, the space complexity is polynomial, and for some function $f$, the running time is $O(f(k) n \log n) = O(f(d) n \log n)$ parametrized by the treewidth $k$ or tree-depth $d$.\\
The essential part is to do the computation on the zeta transformed mirror image of the tree decomposition to save space. 
As in \cite{martin}, it is important to introduce the edges in auxiliary leaves to avoid an exponential blow-up.


\newpage
\appendix

\end{document}